\documentclass[11pt]{article}
 \usepackage{graphicx}
 \usepackage{amssymb}
\usepackage{color}
\usepackage{url}
\usepackage{amsmath}
\usepackage{amsthm}

\newcommand{\bc}{{\bf C}}

\newcommand{\bfx}{{\bf X}}
\newcommand{\mskt}{{MS$k$T} }

\newcommand{\bx}{{\bf X}}

\newtheorem{definition}{Definition}
\newtheorem{proposition}{Proposition}

\newtheorem{corollary}{Corollary}
\newtheorem{theorem}{Theorem}
\newtheorem{lemma}{Lemma}

\begin{document}

\title{Efficient Learning of Optimal Markov Network Topology \\with $k$-Tree Modeling}

\author{Liang Ding\footnote{To whom correspondence should be addressed. Co-authors are listed alphabetically. $^{\, \dagger}$Department of Computer Science, The University of Georgia, Athens, GA; $^{\ddagger}$St. Jude Children's Research Hospital, Memphis, TN; $^{\mathsection}$Department of Plant Biology, The University of Georgia, Athens, GA.}
$^{\, \ddagger}$, Di Chang$^{\,\dagger}$, Russell Malmberg$^{\, \mathsection}$, Aaron Martinez$^{\, \dagger}$,\\ 
David Robinson$^{\, \dagger}$, Matthew Wicker$^{\, \dagger}$, Hongfei Yan$^{\,\dagger}$, and Liming Cai\footnotemark[1]$\mbox{ }^{\, \dagger}$}

\date{\today}
\maketitle

\abstract The seminal work of Chow and Liu (1968) shows that approximation of a finite probabilistic system by Markov trees can achieve the minimum information loss with the topology of a maximum spanning tree. Our current paper generalizes the result to Markov networks of 
tree width $\leq k$, for every fixed $k\geq 2$. In particular, we prove that approximation of a finite probabilistic system with such Markov networks has the minimum information loss when the network topology is achieved with a maximum spanning $k$-tree. While constructing a maximum spanning $k$-tree is intractable for even $k=2$, we show that polynomial algorithms can be ensured by a sufficient condition accommodated by many meaningful applications.  In particular, we prove an efficient algorithm for learning the optimal topology of higher order correlations among random variables that belong to an underlying linear structure. 

\vspace{2mm}


\section{Introduction}
We are interested in effective modeling of complex systems whose the behavior is determined by  relationships among uncertain events. When such uncertain events are quantified as random variables ${\bf X}= \{ X_1, \dots, X_n\}$, the system is characterizable with an underlying joint probability distribution function $P({\bf X})$ \cite{Pearl1988,KollerAndFriedman2010}. Nevertheless, estimation of the function $P({\bf X})$ from observed event data poses challenges as relationships among events may be intrinsic and the $n^{\rm th}$-order  function $P({\bf X})$ may be very difficult, if not impossible, to compute. 
A viable solution is to approximate $P({\bf X})$ with lower order functions, as a result of constraining the dependency relationships among random variables. 
Such constraints can be well characterized with the notion of {\it Markov network}, where the dependency relation of random variables is defined by a non-directed graph $G=({\bf X}, E)$, with the edge set $E$ to specify the topology of the dependency relation. 
In a Markov network, two variables not connected with an edge are independent conditional on the rest of the variables \cite{KindermannAndSnell1980,Pearl1988,RueAndHeld2005}.


Model approximation with Markov networks needs to address two essential issues: the quality of the approximation, and the feasibility to compute such approximated models. Chow and Liu \cite{ChowAndLiu1968} were the first to address both issues in investigating 
networks of tree topology. They measured the information loss with Kullback-Leibler divergence \cite{KullbackAndLeibler1951} between 
$P({\bf X})$, the unknown distribution function,  and $P_G({\bf X})$, the distribution function estimated under the dependency graph $G$. They showed that 
the minimum information loss is guaranteed by the topology corresponding to a maximum spanning tree which can be computed very efficiently (in linear time in the number of random variables). 

It has been difficult to extend the seminal work of
simultaneous optimal and efficient 
learning to the realm of arbitrary Markov networks due to the computational intractability nature of the problem \cite{ChickeringEtAl1994}. To overcome the barrier, research in Markov network learning has sought heuristics algorithms that are efficient yet without an optimality guarantee in the learned topology  \cite{LeeEtAl2006,KoivistoAndSood2004,TeyssierAndKoller2005,KollerAndFriedman2010,DalyEtAl2011,YuanAndMalone2013}. A more viable approach has been to consider the nature of probabilistic networks formulated from real-world applications, which often are constrained and typically tree-like \cite{Dasgupta1999,MeilaAndJordan2000,KargerAndSrebro2001,BachAndJordan2002,Srebro2003,NarasimhanAndBilmes2003,ElidanAndGould2008,BradleyAndGuestrin2010,SzantaiAndKovacs2012,YinEtAl2015}. 
All such models can be quantitatively characterized with networks that have tree width $\leq k$, for small $k$. Tree width is a metric that measures how much a graph is tree-like \cite{Bodlaender2006}. 
It is closely related to the maximum size of a graph separator, i.e., a characteristic of any set of random variables ${\bf X}_S$ upon which two (sets) of random variables ${\bf X}_A$ and ${\bf X}_B$ are conditionally independent ${\bf X}_A \bot {\bf X}_B | {\bf X}_S$. 
Such networks have the advantage to support efficient inference, and other queries over constructed models \cite{KwisthoutEtAl2010}. In addition, many real-world applications actually imply Markov (Bayesian) networks of small tree width. In spite of these advantages, however, optimal learning of such networks is unfortunately difficult, for example, it is intractable even when tree width $k=2$ \cite{KargerAndSrebro2001,Srebro2003}. Hence, the techniques for learning such Markov networks have heavily relied on 
heuristics that may not guarantee the quality. 

In this paper, we generalize the seminal work of Chow and Liu from Markov trees to Markov networks of tree width $\leq k$. First, we prove that model approximation with Markov networks of tree width $\leq k$ has the minimum information loss when the network topology is the one for a maximum spanning $k$-tree. A $k$-tree is a maximal graph of tree width $k$ to which no edge can be added without an increase of tree width. Therefore the quality of such Markov networks of tree width $\leq k$ is also optimal with a maximum spanning $k$-tree.  Second, we reveal sufficient conditions satisfiable by many applications, which guarantee polynomial algorithms to compute the maximum spanning $k$-tree. That is, under such a condition, the optimal topology of Markov networks with tree width $\leq k$ can be learned in time $O(n^{k+1})$ for every fixed $k$.

We organize this paper as follows. Section 2 presents preliminaries and introduces Markov $k$-tree model. In section 3, we prove that Kullback-Leibler divergence $D_{\small\rm KL}(P\parallel P_{\small G})$ is minimized when the approximated distribution function $P_{\small G}$ is estimated with a maximum spanning $k$-tree $G$. In section 4, we present some conditions, each being sufficient for permitting $O(n^{k+1})$-time algorithms to compute maximum spanning $k$-tree and thus to learn the optimal topology of Markov networks of tree width $\leq k$. Efficient inference with Markov $k$-tree models is also discussed in section 4. We conclude in section 5 with some further remarks.

\section{Markov $k$-Tree Model}

The probabilistic modeling of a finite system involves random variables ${\bf X} = \{X_1, \dots, X_n\}$ with observable values.  
Such systems are thus characterized by dependency relations among the random variables, which may be learned by estimating the joint probability distribution function $P({\bf X})$. Inspired by Chow and Liu's work \cite{ChowAndLiu1968}, we are interested in approximating the $n^{\rm th}$-order dependency relation among random variables by a $2^{\rm nd}$-order (i.e., pairwise) relation. Our work first concerns  how to minimize the information loss in the approximation.

A pairwise dependency relation between random variables ${\bf X} = \{X_1, \dots, X_n\}$\footnote{Here we slightly abuse 
notations. We use the same symbol $X_i$ for random variable and its corresponding vertex in the topology graph and rely on context to distinguish them.} can be defined  with a graph $G=(\bx, E)$, where binary relation $E \subseteq \bx \times\bx$ is defined such that $(X_i, X_j) \not \in E$ if and only if variables $X_i$ and $X_j$ are independent conditional only on the rest of the variables, i.e., let 
${\bf Y}={\bf X}\backslash \{X_i, X_j\}$,
\begin{equation}
\label{cond-ind}
 	(X_i, X_j) \not \in E \Longleftrightarrow P_G(X_i, X_j |{\bf Y}) = P_G(X_i |{\bf Y})  P_G(X_j |{\bf Y}) 
\end{equation}
 We call $P_G(\bx)$ a {\it Markov network with the topology graph} $G$ over random variable set 
 ${\bf X}$. 
 
Condition (\ref{cond-ind}) is called the pairwise property. By \cite{KollerAndFriedman2010}, it is equivalent to the following {\it global property} for positive density.  Let ${\bf Y}, {\bf Z} \subseteq \bx$  be two disjoint subsets and ${\bf S} \subseteq \bx $. 
If $\bf S$ separates $\bf Y$ and $\bf Z$ in the graph $G$ (i.e., removing $\bf S$ disconnects $\bf Y$ and $\bf Z$ in the graph), then 
\begin{equation}
\label{global}
	P_G({\bf Y}, {\bf Z} |{\bf S}) = P_G({\bf Y} |{\bf S})  P_G({\bf Z} |{\bf S}) 
\end{equation}
For the convenience of our discussion, we will assume our Markov networks to have the global property.

Equation~(\ref{global}) states that random variables in $\bf Y$ and in $\bf Z$ are independent conditional on variables in ${\bf S}$. The minimum cardinality of $S$ characterizes the degree of condition for the independency between ${\bf Y}$ and ${\bf Z}$. The conditional independency of ${\bf X}$ is thus characterized by the upper bound of minimum cardinalities of all such $\bf S$ separating the graph $G$. This number is closely related to the notion of {\it tree width} \cite{RobertsonAndSeymour1986}. In this paper, we are interested in Markov networks whose topology graphs are of tree width $\leq k$, for given integer $k\geq 1$ \cite{KargerAndSrebro2001}.
 

\subsection{$k$-Tree and Creation Order}
\label{ktree-section}

Intuitively, tree width \cite{Bodlaender2006} is a metric for how much a graph is tree-like. It plays one of the central roles in the development of structural graph theory and, in particular, in the well known graph minor theorems \cite{RobertsonAndSeymour1986} by Robertson and Seymour. There are other alternative definitions for tree width. For the work presented in this paper, we relate its definition to the concept of $k$-tree.

\begin{definition}
\label{def_ktree}
 \rm
\cite{Patil1986} Let $k\geq 1$ be an integer. The class of {\it $k$-trees} are graphs that are defined inductively as follows:
\vspace{0mm}
\begin{enumerate}
\vspace{0mm}
\item A $k$-tree of $k$ vertices is a clique of $k$ vertices;
\vspace{-1mm}
\item A $k$-tree $G$ of $n$ vertices, for $n > k$, is a graph consisting of a $k$-tree $H$ of $n-1$ vertices and a new vertex $v$ not in $H$, such that $v$ forms a $(k+1)$-clique with some $k$-clique already in $H$.
\end{enumerate}
\end{definition}

\begin{definition}
\rm 
\cite{Bodlaender2006} For every fixed $k\geq 1$, a graph is said to be of {\it tree width} $\leq k$ if and only if it is a subgraph of some $k$-tree. 
\end{definition}


In particular, all trees are $k$-trees for $k=1$ and they have tree width 1. Graphs of tree width 1 are simply forests. 
For $k\geq 1$, a $k$-tree can be succinctly represented with the vertices in the order in which they are introduced to the graph by the inductive process in Definition 1. 

\begin{definition} \rm
Let $k\geq 1$ be an integer.
A {\it creation order} \, $\Phi_{\bx}$ for a $k$-tree $G=(\bx, E)$, where 
$\bx=\{X_1, \dots, X_n\}$, is  defined  as an {\it ordered sequence}, inductively 
\begin{equation}
\label{creation_order}
\Phi_{\bx} = \begin{cases} \bc \hspace{20mm} & \,  |\bx| =k, \mbox{ where } \bc = \bx\\
\bc X & \, |\bx| = k+1, \mbox{ where } \bc \cup \{X\} = \bx\\
 \Phi_{\bx \backslash\{X\}}\, \bc X & \, |\bx| >k+1, \mbox{ where }
 X \not \in \bc \subset \bx\\
\end{cases}
\end{equation}
where $\bc$ is a $k$-clique and $\bc \cup \{X\}$ is a $(k+1)$-clique. 
\end{definition}

Note that creation orders may not be unique for the same $k$-tree. The succinct notation facilitates discussion of properties for $k$-trees. In particular, because Markov networks are non-directed graphs, the notion of creation order offers a convenient mechanism to discuss dependency between random variables modeled by the networks. 

\begin{proposition}
\label{prefix}
Let $\Phi_{\bx}$ be a creation order for $k$-tree $G=(\bx, E)$. Then $\kappa$ is a $(k+1)$-clique in $G$ if and only if there are $X\in \kappa$ and ${\bf Y} \subseteq \bx\backslash\{X\}$ such that $\Phi_{\bf Y} \bc X$ is a prefix of $\Phi_{\bx}$, where $\bc=\kappa\backslash\{X\}$.
\end{proposition}

\begin{definition}
\label{orientation}
\rm Let $k\geq 1$ be an integer and $\Phi_{\bx}$
be a creation order of some $k$-tree $G=(\bx, {E})$, where $\bx=\{X_1, \dots, X_n\}$.
Then $\Phi_{\bx}$ {\it induces an acyclic orientation} $\hat{E}$ for $E$ such that given any two different variables $X_i, X_j \in \bx$, 
\begin{equation}
\label{parentcases}
 \langle X_i, X_j\rangle \in \hat{E} \, \mbox{ if and only if }  \,
 \begin{cases} 
 X_i, X_j \in \bc, \,  i < j, \, \mbox{ and } \, \bc \mbox{ is a prefix of } \Phi_{\bx}\\
 X_i \in \bc \mbox{ and } \Phi_{\bf Y} \bc X_j \mbox{ is a prefix of } \Phi_{\bx}, \mbox{ for some } {\bf Y}
 \end{cases}
\end{equation}
In addition, for any variable $X \in \bx$, its {\it parent set} is defined as $\pi_{\hat{E}}(X) = \{ X_i: \,  \langle X_i, X\rangle \in \hat{E}\}$
\end{definition}
Based on Definition~\ref{orientation}, it is clear that
the orientation induced by any creation order of $k$-tree is acyclic.
In particular, there exists exactly one variable $X$ with $\pi_{\hat{E}}(X) = \emptyset$.

\subsection{Relative Entropy and Mutual Information}

According to Shannon's information theory \cite{Shannon1948}, 
the uncertainty of a discrete random variable $X$ can be measured with {\it entropy} $H(X)=   - \sum\limits_{x} P(X)\log_2 P(X)$, where $P(X)$ is the probability distribution for $X$ and the sum takes all values $x$ in the range of $X$. Entropy for a set of random variables $\bx$ is  $H(\bx)=   - \sum\limits_{\bf{x}} P(\bx)\log_2 P(\bx)$, where the sum takes all combined values $\bf x$ in the ranges of variables in $\bx$.
\begin{definition}
\rm \cite{KullbackAndLeibler1951}
Kullback-Leibler {\it divergence} between two probability distributions $P(\bx)$ and $Q(\bx)$ of the same random variable set $\bx$ is defined as
\begin{equation}
\label{KL}
D_{KL}(P\parallel Q) = \sum_{{\bf x}} P(\bfx) \log \frac{P(\bfx)}{Q(\bfx)} \geq 0
\end{equation}
where the sum takes combined values $\bf x$ in the ranges of all random variables in $\bx$. 
\end{definition}
The last equality holds if and only if $P(\bfx) = Q(\bfx)$. 
$D_{KL}(P\parallel Q) $ can be used to measure the information gain by distribution $P$ over distribution $Q$, or information loss by approximation of $P$ with $Q$ as in this work. 

Let $X$ and $Y$ are two random variables. The {\it mutual information} between $X$ and $Y$, denoted as $I(X; Y)$, is defined as the Kullback-Leibler divergence between their
joint probability distribution and the product of marginal distributions, i.e., 
\[ I(X; Y) = D_{KL}(p(X, Y)\parallel p(X)p(Y)) = \sum_{x,y} p(X, Y) \log \frac{p(X,Y)}{p(X)p(Y)}\]
Mutual information $I(X; Y)$ measures the degree of the correlation between the two random variables. 
In this work, we slightly extends the mutual information to include more than two variables. 
\begin{definition}
\rm
Let $X$ be a random variable and $\bc$ be a set of random variables. Then 
\[ I(X; \bc) = D_{KL}(p(X, \bc)\parallel p(X)p(\bc)) = \sum_{x, {\bf c}} p(X, \bc) \log \frac{p(X,\bc)}{p(X)p(\bc)}\]
where sum takes combined values ${x}$ and ${\bf c}$ in the ranges of $X$ and all random variables in $\bc$. 
\end{definition} 
Note that $I(X; \bc)$ is not the same as the multivariate mutual information defined elsewhere \cite{TimmeEtAl2014}. We further point out that mutual information has recently received much attention due to its capability in decoding hidden information from big data \cite{KinneyAndAtwal2014,ZhaoEtAl2015}. In our derivation of optimal Markov $k$-tree modeling, estimating mutual information stands out as a critical condition that needs to be satisfied (see next section). These phenomena are unlikely coincidences.

\subsection{Markov $k$-Tree Model}

\begin{definition}
\label{def_markov_ktree}
\rm Let $k\geq 1$ be an integer. A {\it Markov $k$-tree} is a Markov network over $n$ random variables $\bx =\{X_1, \dots, X_n\}$ with a topology graph $G=(\bx, E)$ that is a $k$-tree. We denote with $P_G(\bx)$ the joint probability distribution function of the Markov $k$-tree. \end{definition}

\vspace{2mm}

\begin{theorem} Let $k\geq 1$ and ${\bf X} =\{X_1, \dots, X_n\}$ be random variables with joint probability distribution function $P(\bx)$. Let $G=(\bx, E)$ be a Markov $k$-tree and 
$\hat{E}$ be an acyclic orientation for its edges. 
Then 
\begin{equation}
\label{joint_prob}
{P_G(\bfx)  =  \prod_{i=1}^{n} P(X_i | \pi_{\hat{E}}(X_i))} 
\end{equation}
\end{theorem}

\begin{proof}

Assume that the acyclic orientation $\hat{E}$ is induced by creation order 
\[\Phi_{\bx} = \bc_{k} X_{k+1} \dots \bc_{n-1} X_{n}\] where $\bc_{k} = \{ X_1, \dots, X_k\}$
\begin{equation}
\label{PG}
\begin{split}
 P_G(\bx) & = P(\bx | \Phi_{\bx}) \\
 & = P(X_1, \dots, X_n | \Phi_{\bx}) \\
 & = P(X_n | X_1, \dots, X_{n-1}, \Phi_{\bx})  P(X_1, \dots, X_{n-1} | \Phi_{\bx}) \\
 & = P(X_n | C_{n-1}) P(\bx\backslash\{X_n\}) | \Phi_{\bx \backslash X_n})
 \end{split}
 \end{equation}
The last equality holds for the reason that $X_n$ is independent from variables in $\bx\backslash \bc_{n-1}$ conditional
on variables in $\bc_{n-1}$, as shown in the following derivation. Assume ${\bf Y}=\bx \backslash (\bc_{n-1} \cup \{X_n\})$. Use the Bayes, 
\[ P(X_n | X_1, \dots, X_{n-1}, \Phi_{\bx}) = P(X_n | {\bf Y}, \bc_{n-1}, \Phi_{\bx})
=\frac{P(X_n, {\bf Y} | \bc_{n-1}, \Phi_{\bx})} { P({\bf Y} | \bc_{n-1}, \Phi_{\bx})}\]
\[ 
=\frac{P(X_n | \bc_{n-1}, \Phi_{\bx}) P({\bf Y} |\bc_{n-1}, \Phi_{\bx}) } { P({\bf Y} | \bc_{n-1}, \Phi_{\bx})} = P(X_n | \bc_{n-1}, \Phi_{\bx}) 
= P(X_n | \bc_{n-1})\]

 The derivation in (\ref{PG}) also results in recurrence
\begin{equation}
\label{recurrence}
 P(\bx | \Phi_{\bx}) =
 \begin{cases} 
 P(X_n | \bc_{n-1}) P(\bx\backslash \{X_n\} | \Phi_{\bx\backslash \{X\}}) & \mbox{ if } n > k\\
  P(\bc_{k})  & \mbox{ if } n = k
 \end{cases}
 \end{equation}
  Solving the recurrence yields  
 \begin{equation}
 \label{creation-order-prob}
 P(\bx | \Phi_{\bx}) = P(\bc_{k}) \prod\limits_{i=k+1}^n P(X_{i-1} | \bc_{i})
 \end{equation}
where $P(\bc_{k}) = P(X_1, \dots, X_k)$. Because for $i\leq k$, $\pi_{\hat{E}}(X_i) = \{X_1, \dots, X_{i-1}\}$, it is not hard to prove that 
\begin{equation}
\label{ck}
P(X_1, \dots, X_k) = \prod\limits_{i=1}^k P(X_i | \pi_{\hat{E}}(X))
\end{equation} 

In addition, because 
$\pi_{\hat{E}}(X_i) = \bc_{i-1}$ for all $i\geq k+1$, by equations 
(\ref{creation-order-prob}) and (\ref{ck}), we have 
\[P_G(\bx) = \prod_{i=1}^k P(X_i | \pi_{\hat{E}}(X_i))  \prod_{i=k+1}^n P(X_i | \pi_{\hat{E}}(X_i)) = \prod_{i=1}^n P(X_i | \pi_{\hat{E}}(X_i))\]
\end{proof}
Though the above proof is based on an explicit creation order for the $k$-tree $G$. The probability function computed for $P_G(\bx)$  is actually independent of the choice of a creation order as demonstrated in the following.

\begin{theorem}
\label{creation-order-independent}
The probability function expressed in {\rm (\ref{joint_prob})} for Markov $k$-tree $G$ remains the same regardless the choice of creation order for $G$.
\end{theorem}

A proof is given in Appendix A (Theorem~\ref{appendix}).



\section{Model Optimization}

\subsection{Optimal Markov $k$-Trees}

Let $\bx =\{X_1, \dots, X_n\}$ be random variables in any 
unconstrained probabilistic model $P(\bx)$. Let $k\geq 1$ and $k$-tree $G=(\bx, E)$ be a topology graph of a Markov network $P_G(\bx)$ that approximates  
$P(\bx)$. The approximation can be measured using the Kullback-Leibler divergence between two models \cite{KullbackAndLeibler1951,Lewis1959}:
\begin{equation}
\small
\label{KL_divergence}
 D_{KL}(P\parallel P_G) = \sum_{x} P(\bfx) \log \frac{P(\bfx)}{P_G(\bfx)} \geq 0
\end{equation}
where $x$ is any combination of values for all random variables in $\bx$, and the last equality holds if and only if $P(\bfx) = P_G(\bfx)$. 

Thus the problem to optimally approximate 
$P(\bx)$ is to find a Markov $k$-tree with topology $G$ that minimizes the divergence (\ref{KL_divergence}). 
We are now ready for our first main theorem.

\begin{theorem} 
\label{min-DKL}
$D_{KL}(P\parallel P_G)$ is minimized when the topology $k$-tree
$G=(\bx, E)$ for random variables $\bx =\{X_1, \dots, X_n\}$ is such 
that maximizes the sum of mutual information.
\begin{equation}
\label{max-delta}
\sum\limits_{i=1}^n I(X_i, \pi_{\hat{E}}(X_i)) 
\end{equation}
where $\hat{E}$ is any acyclic orientation for the edges in $G$.
\end{theorem}

\begin{proof} 

Assume $\hat{E}$ to be an acyclic orientation for the edges in $G$. 
Apply equation (\ref{joint_prob}) 
\[P_G(\bfx)  =  \prod_{i=1}^{n} P(X_i | \pi_{\hat{E}}(X_i))\]
to $P_G(\bx)$ in (\ref{KL_divergence}), we have
\[  D_{KL}(P\parallel P_G) = \sum\limits_{x} P(\bfx) \log P(\bfx) - \sum_{x} P(\bfx) \sum_{i=i}^n \log P(X_i | \pi_{\hat{E}}(X_i))
\]
The first term on the right hand side (RHS) of above is  $-H(\bfx)$, where $H(\bx)$ is the entropy of \bfx. And the second term 
$- \sum\limits_{x} P(\bfx) \sum\limits_{i=1}^n \log P(X_i | \pi_{\hat{E}}(X_i))$ can be further explained with 
\begin{equation}
\label{more_terms}
 - \sum_{x}  P(\bfx) \sum_{i=1}^n \log \frac{P(X_i, \pi_{\hat{E}}(X_i))}{P(X_i) P(\pi_{\hat{E}}(X_i))} - \sum_{x}  P(\bfx) \sum_{i=1}^n \log P(X_i) 
\end{equation}
Because $P(X_i)$ and $P(X_i, \pi_{\hat{E}}(X_i))$ are components of $P(\bfx)$, the second term in (\ref{more_terms}) can be formulated as
\begin{equation}
\begin{split} - \sum_{x}  P(\bfx) \sum_{i=1}^n \log P(X_i)
 & =  \sum_{i=1}^n \sum_{x}  P(\bfx)  \log P(X_i) \\
 & = \sum_{i=1}^n \sum_{x_i} P(X_i) \log P(X_i) \\
 & = \sum\limits_{i=1}^n H(X_i)\\
\end{split}
 \end{equation}
where $x_i$ is any value in the range of random variable $X_i$.

Also the first term in (\ref{more_terms}) gives 
\begin{equation}
\begin{split}
 -\sum_{i=1}^n \sum_{x} P(\bfx) \log \frac{P(X_i, \pi_{\hat{E}}(X_i))}{P(X_i) P(\pi_{\hat{E}}(X_i))} 
 & =  - \sum_{i=1}^n \sum_{x_i, y_i} P(X_i, \pi_{\hat{E}}(X_i)) \log \frac{P(X_i, \pi_{\hat{E}}(X_i))}{P(X_i) P(\pi_{\hat{E}}(X_i))} \\
 &  =  - \sum_{i=1}^n I(X_i, \pi_{\hat{E}}(X_i))
 \end{split}
 \end{equation}
where $y_i$ is any combination of values of all random variables in $\pi_{\hat{E}}(X_i)$ and $I$ is the mutual information between variable $X_i$ and its parent set $\pi_{\hat{E}}(X_i)$.

Therefore, 
\begin{equation}
\label{opt}
 D_{KL}(P\parallel P_G) 
= - \sum_{i=1}^n I(X_i, \pi_{\hat{E}}(X_i)) + \sum\limits_{i=1}^n H(X_i) - H(\bfx) 
\end{equation}

Since $\sum\limits_{i=1}^nH(X_i)$ and $H(\bfx)$ are independent of the choice $G$ (and the acyclic 
orientation for the edges), $D_{KL}(P\parallel P_G)$ is minimized when $\sum\limits_{i=1}^n I(X_i, \pi_{\hat{E}}(X_i))$ is maximized.
\end{proof}

Though Theorems~\ref{min-DKL} is about optimal approximation of probabilistic systems with Markov $k$-trees, they also characterize the optimal Markov $k$-tree. 
\begin{theorem} 
\label{opt-markov-ktree}
Let $k\geq 1$ and $\bx=\{X_1, \dots, X_n\}$ be a set of random variables. The optimal Markov $k$-tree model for $\bx$ is the one with topology $G=(\bx, E)$ that maximize $\sum\limits_{i=1}^n I(X_i, \pi_{\hat{E}}(X_i))$, for some acyclic orientation of edges in $E$.
\end{theorem}

Now we let 
\[
 \Delta_{G,\hat{E}} (\bx) = \sum\limits_{i=1}^n I(X_i, \pi_{\hat{E}}(X_i))
 \]
By Theorem~\ref{creation-order-independent} and the proof of Theorem~\ref{min-DKL}, we know that 
$\Delta_{G,\hat{E}}$ is invariant of the choice of an acyclic orientation $\hat{E}$ of edges. So we can simply omit $\hat{E}$ and use $\Delta_{G} (\bx)$ for $\Delta_{G,\hat{E}}(\bx)$.
\begin{definition}
\label{def-G-star}
\rm Let $\bx =\{X_1, \dots, X_n\}$ be random variables. Define 
\begin{equation}
\label{G-star}
G^* 
  = \arg \max_{G} \{ \Delta_{G}(\bx)\}
 \end{equation}
 to be the topology of optimal Markov $k$-tree over $\bx$.
\end{definition}

\begin{corollary}
Let $k\geq 1$ and 
$\bx=\{X_1, \dots, X_n\}$ be a set of random variables. Divergence $D_{KL}(P\parallel P_{G^*})$ is minimized with the topology $k$-tree ${G}^*$ satisfying {\rm (\ref{G-star})}.  
\end{corollary}  
 
\begin{corollary} 
Let $k\geq 1$ and $\bx=\{X_1, \dots, X_n\}$ be a set of random variables. The optimal Markov $k$-tree model for $\bx$ is the one with topology graph $G^*=(\bx, E)$ satisfying {\rm (\ref{G-star})}.
\end{corollary}


\subsection{Optimal Markov Networks of Tree Width $\leq k$}

Because $k$-trees are the maximum graphs of tree width $\leq k$, Theorems~\ref{min-DKL} and \ref{opt-markov-ktree} are not immediately applicable to optimization of Markov networks of tree width $\leq k$. Our derivations so far have not included   
the situation that, given $k\geq 1$, an optimal
 Markov network of tree width $\leq k$ is not exactly a $l$-tree, for any $l\leq k$.  
In this section, we resolve this issue with a slight adjustment to the objective function $\Delta_{G}(\bx)$.


\begin{definition}
\label{amendment}
\rm
Let $G=(\bx, E)$ be a $k$-tree and $A \subseteq E$ be a binary relation over $\bx$. 
The spanning subgraph $G_A=(\bx, A)$ is called {\it amended graph} of $G$ with $A$.
\end{definition}

\begin{proposition}
A graph has tree width $\leq k$ if and only if it is an amended graph $G_A=(\bx, A)$  of some $k$-tree $G=(\bx, E)$ with  relation $A$.

\end{proposition} 

\begin{definition}\rm 
\label{ammended-orientation}
Let $\hat{E}$ be an acyclic orientation for edges in $k$-tree $G=(\bx, E)$ and $G_A=(\bx, A)$ be an amended graph
of $G$. 
Then the orientation $\hat{A}$ of edges in the graph $G_{A}$ is defined as, for every pair of $X_i, X_j \in \bx$, 
\[ \langle X_i, X_j \rangle \in \hat{A}  \, \Longleftrightarrow  \, \langle X_i, X_j \rangle \in \hat{E} \, \wedge \, (X_i, X_j) \in A\]
\end{definition}

\begin{definition}
\label{new-delta}\rm 
Let $G=(\bx, E)$ be a $k$-tree, where $|\bx|=n$.
Let $G_A=(\bx, A)$ be an amended graph of $G$. Define
\begin{equation}
\label{obj3}
\Delta_{G, \hat{A}}(\bx) = \sum_{i=1}^n I(X_i, \pi_{\hat{A}}(X_i))
\end{equation}
where $\hat{A}$ is an acyclic orientation on edges $A$.
\end{definition}

We can apply the derivation in the previous section by replacing $\pi_{\hat{E}}(X_i)$ with
$\pi_{\hat{A}}(X_i)$ and obtain
\begin{corollary} 
\label{opt-Markov-network}
Let $k\geq 1$ and $\bx=\{X_1, \dots, X_n\}$ be a set of random variables. The optimal Markov network of tree width $\leq k$
for $\bx$ is the amended graph $G^*_A=(\bx, A)$ of some $k$-tree $G=(\bx, E)$ 
that satisfies 
\[ G_A^* = \arg \max_{G_A, \hat{A}} \{ \Delta_{G, \hat{A}}(\bx) \}\]
\end{corollary}

\section{Efficient Computation with Markov $k$-Trees}

Optimal topology learning is computationally intractable for Markov networks of arbitrary graph topology \cite{ChickeringEtAl1994}. This obstacle is not exception to Markov $k$-trees. In particular, we are able to relate the optimal learning of Markov $k$-trees to the following graph-theoretic problem. 

\begin{definition} 
\label{mskt}
\rm Let $k\geq 1$. The {\sc Maximum Spanning $k$-Tree} (MS$k$T) problem is, on input graph of vertices 
$\bx=\{X_1, \dots, X_n\}$, to find a spanning $k$-tree $G=(\bx, E)$ with an acyclic orientation $\hat{E}$ 
such that 
\[ \sum\limits_{i=1}^n f(X_i, \pi_{\hat{E}}(X_i))\]
achieves the maximum.
\end{definition}

In the definition, function $f$ is pre-defined numerical function for a pair: vertex $X_i$ and its parent set $\pi_{\hat{E}}(X_i)$ in the output $k$-tree. MS$k$T generalizes the traditional ``standard'' definition of the maximum spanning $k$-tree problem \cite{Bern1987,CaiAndMaffray1993}, where the objective function is the sum of weights $w: E \rightarrow R$ on all edges involved the output $k$-tree; that is  
\begin{equation}
\label{sum_of_edge_weights}
  f(X_i, \pi_{\hat{E}}(X_i)) =   \sum_{X\in \pi_{\hat{E}}(X_i)} w(X, X_i)
\end{equation}

\begin{proposition}
\label{markov-mskt}
The Markov $k$-tree learning problem is the problem {\rm MS$k$T} defined in {\rm Definition~\ref{mskt}} in which
for every $i=1,\dots, n$, $f(X_i, \pi_{\hat{E}}(X_i)=I(X_i, \pi_{\hat{E}}(X_i))$, the mutual information between $X_i$ and its parent set $\pi_{\hat{E}}(X_i)$.
\end{proposition}

It has been proved that, for any fixed $k\geq 2$, the problem {\sc MS$k$T} with the objective function defined with $f$ given in equation (\ref{sum_of_edge_weights}) is NP-hard \cite{Bern1987}. The intractability appears inherent since a number of variants of the problem remains intractable \cite{CaiAndMaffray1993}. It implies that 
optimal topology learning of Markov $k$-trees (for $k\geq 2$) is computationally intractable. Independent works in learning Markov networks of bounded tree width have also confirmed this unfortunate difficulty \cite{Srebro2003,SzantaiAndKovacs2012}.


\subsection{Efficient Learning of Optimal Markov Backbone $k$-Trees}

We now consider a class of less general Markov $k$-trees that are of higher order topologies 
over an underlying linear relation. In particular, such networks carry the signature of relationships among consecutively indexed random variables. 

\begin{definition}
\label{backbone}
\rm
Let $\bx = \{X_1, X_2, \dots, X_n\}$ be the set of integer-indexed variables. 
\vspace{0mm}
\begin{itemize}
\item[(1)] The graph
$B_{\bx}=(\bx, b)$, where $b =\{ (X_i, X_{i+1}): \, 1\leq i < n\}$, is called the {\it backbone} of $\bx$. And each edge $(X_i, X_{i+1})$, for $1\leq i < n$, is called a {\it backbone edge}.
\vspace{-1mm}
\item[(2)] A graph (resp. $k$-tree) $G=(\bx, E)$ is called a 
{\it backbone graph} (resp. {\it backbone $k$-tree}) if it contains $B_{\bx}$ as a subgraph.
\end{itemize} 
\end{definition}


\begin{definition} \rm
Let $k\geq 1$. A Markov network is called {\it Markov backbone $k$-tree} if the underlying topology graph of the Markov network is a backbone $k$-tree. 
\end{definition}

The linearity relation among random variables occurs naturally in many applications with Markov network modeling, for instance, in speech recognition, cognitive linguistics, and higher-order relations among residues on biological sequences. Typically, the linearity also plays an important role in random systems involving variables associated with the time series. For example, 

\begin{theorem}
For each $k\geq 1$, the topology graph of any finite $k$th-order Markov chain is a backbone $k$-tree. 
\end{theorem}

\begin{proof}
This is because a finite $k$th-order Markov chain $G$ is defined over labelled variables $X_1, \dots, X_n$ such that $P_G(X_i | X_1, \dots, X_n) = P_G(X_i | X_{i-1}, \dots, X_{i-k})$, for every $i=1, 2, \dots, n$. Clearly, its topology graph contains all edges $(X_i, X_{i+1})$, for $i, 1\leq i \leq n-1$. In addition, the following creation order asserts that the graph is a $k$-tree: $\bc_k X_{k+1} \dots \bc_{i-1} X_{i} \dots \bc_{n-1} X_n$, where $\bc_{i-1} = \{X_{i-k}, \dots, X_{i-1}\}$, for all $i= k+1, \dots, n$.
\end{proof}

We now show that learning the optimal topology of Markov backbone $k$-trees can be accomplished much more efficiently than with unconstrained Markov $k$-trees.
We relate the learning problem to a special version of graph-theoretic problem {\sc MS$k$T} that can be computed efficiently. In the rest of this section, we present some technical details for such connection.

\begin{definition}
\label{def_condition} \rm Let ${\cal G}$ be a class of graphs. 
Then the {\it $\cal G$-retaining} \mskt problem is the \mskt problem in which 
the input graph includes a spanning subgraph $H \in {\cal G}$ that is required to be contained by the output spanning $k$-tree.
\end{definition}

The unconstrained \mskt is simply the ${\cal G}$-retaining \mskt problem with ${\cal G}$ being the class of independent set graphs. There are other classes of graphs that are of interest to our work.  In particular, we define $\cal B$ to be the class of all backbones. Then it is not difficult to see:

\begin{proposition}
\label{markov-msbkt}
The Markov backbone $k$-tree learning problem is the ${\cal B}$-retaining {\rm \mskt}problem.
\end{proposition}

In the following, we will show that,
for various classes ${\cal G}$ that satisfy certain graph-theoretic property, including class ${\cal B}$, the ${\cal G}$-retaining \mskt problems  can be solved efficiently.

\begin{definition}
\label{bb}\rm
Let fixed $k\geq 1$. 
A $k$-tree $G$ has {\it bounded branches} if every $(k+1)$-clique in $G$ separates $G$ into a bounded number of connected components.
\end{definition}

\begin{definition}
\label{bbf}\rm
Let fixed $k\geq 1$. A graph $H$ is {\it bounded branching-friendly} if every $k$-tree that contains $H$ as a spanning subgraph has bounded branches.
\end{definition}

\begin{lemma}
\label{backbone_is_bb_friendly} 
Let $\cal B$ be the class of all backbones.
Any graph in ${\cal B}$ is bounded branching friendly.
\end{lemma}

\begin{proof}

Let $G=(\bx, E)$ be a $k$-tree containing a backbone $B_{\bx}=(\bx, b)$. 
Let $(k+1)$-clique $\kappa=\{x_1, \dots, x_{k+1}\}$ with $x_1< x_2 \dots < x_{k+1}$. Let 
\[{\cal K} =\{ \kappa':\, \kappa' \mbox{ is a  $(k+1)$-clique and } |\kappa \cap \kappa' | = k\}\]
It is not difficult to see that $\kappa$ separates the $k$-tree into at most $|{\cal K}|$ connected components.

Note that $\kappa$ divides set $[n]$ into at most $k+2$ intervals. Let $\kappa_1, \kappa_2\in {\cal K}$  be two $(k+1)$-cliques such that vertices $z_1 \in  \kappa_1 \backslash \kappa$ and $z_2 \in  \kappa_2 \backslash \kappa$. Then $z_1$ and $z_2$ cannot belong to the same interval in order to guarantee all backbone edges are included in $G$. This implies $|{\cal K}| \leq k+2$, a constant when $k$ is fixed.
\end{proof}

\begin{lemma} 
\label{generalmskt}
Let $\cal G$ be a class of bounded branching friendly graphs. Then {\it ${\cal G}$-retaining} {\sc MS$k$T} problem can be solved in time $O(n^{k+1})$ for every fixed $k\geq 1$. 
\end{lemma}

\begin{proof} 

(We give a sketch for proof.)

  Note that the {\it ${\cal G}$-retaining} {\sc MS$k$T} problem is defined as follows:

\vspace{2mm}

$\mbox{ }$\, \, Input: a set of vertices $\bx =\{X_1, \dots, X_n\}$ and a graph $H=(\bx, E_H) \in {\cal G}$, 

$\mbox{ }$\, \, Output: $k$-tree $G^*=(\bx, E)$ with $E_H \subseteq E$ such that  
\[G^* = \arg \max\limits_{G, \hat{E}} \, \{ \, \sum\limits_{i=1}^{n} f(X_i, \pi_{\hat{E}}(X_i)\}\]

$\mbox{ }$\hspace{15mm} \, where $f$ is a pre-defined function.

\vspace{2mm}

 First, any $(k+1)$-clique in a $k$-tree that contains $H$ as a spanning subgraph can only separate $H$ into a bounded number of connected components. To see this, assume some $(k+1)$-clique 
 $\kappa$ does the opposite. From the graph $H$ and $\kappa$ we can construct a $k$-tree by augmenting $H$ with additional edges. We do this without connecting any two of the components in $H$ disconnected by $\kappa$.  However, the resulted $k$-tree, though containing $H$ as a spanning subgraph, would be separated by $\kappa$ into an unbounded number of components, contradicting the assumption that $H$ is bounded branching friendly.  
 
 Second, consider a process to construct a spanning $k$-tree anchored at any fixed $(k+1)$-clique $\kappa$. Note that by the aforementioned discussion, in the $k$-tree (to be constructed) the number of components separated by $\kappa$ is some constant $c(k) \geq 1$, possibly a function in $k$ which is fixed. For every component, its connection to $\kappa$ is by some other $(k+1)$-clique $\kappa'$ such that 
 $\kappa' = \kappa \cup \{x\} \backslash \{y\}$, where $y$ is drawn from the component to ``swap'' out $x$, a vertex chosen from $\kappa$. It is easy to see that, once $x$ is chosen, there are only a bounded number of ways (i.e,, at most $k+1$ possibilities) to choose $y$. On the other hand, there are at most $O(n)$ possibilities to choose $x$ from one of the $c(k)$ connected components. If we use 0 and 1 to indicate the permission mode `yes' or 'no' that $x$ can be drawn from a specified component, an indicator for such permissions for all the components can be represented with $c(k)$ binary bits. 

Let $\alpha$ be a binary string of length $c(k)$. We define $M(\kappa, \alpha)$ to be the maximum value of a sub-$k$-tree created around $\kappa$, which is consistent with the permission indicator $\alpha$ for the associated components. We then have the following recurrence for function $M$: 
\[ M(\kappa, \alpha) = \max_{x, y, \beta, \gamma} \{ M(\kappa', \beta) + M(\kappa, \gamma) + f(x, \kappa\backslash \{y\}) \}\]
where $\kappa' = \kappa \cup \{x\} \backslash  \{y\} $ and $f$ is the function defined over the newly created vertex $x$ and its parent set $\kappa\backslash \{y\}$. The recurrence defines $M(\kappa, \alpha)$ with two recursive terms $M(\kappa', \beta)$ and $M(\kappa, \gamma)$, where the permission indicators $\beta$ and $\gamma$ are such that the state corresponding to the component from which $x$ was drawn should be set '1' by $\beta$ and '0' by $\gamma$. When a component runs out of vertices, the indicator should be set '0'. For other components, where $\alpha$ indicates '1', either $\beta$ or  $\gamma$ indicates '1' but not both. Finally, base case for the recurrence is that $M(\kappa, \alpha)=0$ when $\alpha$ is all 0's.

The recurrence facilitates a dynamic programming algorithm to compute $M$ for all entries of $(k+1)$-clique $\kappa$ and information $\alpha$. It runs in time $O(n^{k+2})$ and uses $O(n^{k+1})$ memory which can be improved to $O(n^{k+1})$-time and $O(n^k)$ memory with a more careful implementation of the idea, e.g., by considering $\kappa$ a  $k$-clique  instead of a $(k+1)$-clique \cite{DingEtAl2014,DingEtAl2015}. 
\end{proof}
Combining Proposition~\ref{markov-msbkt} and Lemmas~\ref{backbone_is_bb_friendly} and \ref{generalmskt}, we have
\begin{theorem}
\label{learning_markov_backbone}
For every fixed $k\geq 1$, optimal topology learning of Markov backbone $k$-trees can be accomplished in time $O(n^{k+1})$.
\end{theorem}

\subsection{Efficient Inference with Markov $k$-Trees}
On learned or constructed Markov $k$-trees, inferences and queries of probability distributions can be conducted, typically through computing the maximum  posterior probabilities for most probably explanations \cite{KollerAndFriedman2010}. The topology nature of Markov $k$-trees offers the advantage of achieving (high) efficiency for such inference computations, especially when  values of $k$ are small or moderate.    

Traditional methods for inference with constructed Markov networks seek additional independency properties that the models may offer. Typically, the method of clique factorization measures the joint probability $P(\bx)$ of a Markov network  in terms of ``potential'' $\phi$ such that 
it can be factorized according to the maximal cliques in the graph: $P(\bx) = \prod_{\kappa \in {\cal K}} \phi_{\kappa}(\bx_{\kappa})$, where ${\cal K}$ is the set of maximal cliques and $\bx_{\kappa}$ is the subset of variables associated with clique $\kappa$. \cite{Bishop2006}. For general graphs, by the Hammersley-Clifford theorem \cite{Grimmett1973}, the condition to guarantee the equation is non-zero probability $P(\bx) > 0$, that may not always be satisfied. For Markov $k$-trees, such factorization suits well because $k$-trees are naturally chordal graphs. Nevertheless, inference with a Markov network may involve tasks more sophisticated than computing joint probabilities. 

The $k$-tree topology offers a systematic strategy to compute many optimization functions over the constructed network models. This is briefly explained below.

Every $k$-tree of $n$  vertices, $n > k$, consists of exactly $n-k$ number of $(k+1)$-cliques. 
By Proposition~\ref{prefix} in Section 2, if a $k$-tree is defined by 
creation order $\Phi_{\bx}$, all  $(k+1)$-cliques uniquely correspond prefixes of 
$\Phi_{\bx}$. Thus relations among the $(k+1)$-cliques can be defined by relations among the prefixes of $\Phi_{\bx}$.
For the convenience of discussion, with $\kappa_{\bc X}$ we denote the $(k+1)$-clique corresponding to prefix $\Phi_{\bf Y} \bc X$ of $\Phi_{\bx}$. 

\begin{definition}
\label{parent_relation}
\rm Let $G=(\bx, E)$ be a $k$-tree and $\Phi_{\bx}$ be any creation order for $G$. 
Two $(k+1)$-cliques $\kappa_{\bc X}$ and $\kappa_{\bc'X'}$  in $G$ are related, denoted
with ${\cal E}(\kappa_{\bc X}, \kappa_{\bc' X'})$, if and only if 
their corresponding prefixes $\Phi_{\bf Y} \bc X$ and $\Phi_{\bf Y'} \bc' X'$ of $\Phi_{\bx}$ satisfy

\vspace{1mm}

(1) $\bc \subset \bc' \cup \{X'\}$; 

(2) $\Phi_{\bf Y'} \bc' X'$ is a prefix of $\Phi_{\bf Y}\bc X$;

(3) No other prefix of $\Phi_{\bf Y}\bc X$ satisfies (1) and of which $\Phi_{\bf Y'} \bc' X'$ is a prefix.
\end{definition}

\begin{proposition}
\label{td}
Let ${\cal K}$ be the set of all $(k+1)$-cliques in a $k$-tree and ${\cal E}$ be the relation of cliques given in Definition~\ref{parent_relation} for the cliques. Then $({\cal K}, {\cal E})$ a rooted, directed tree with directed edges from child $\kappa_{\bc X}$ to its parent $\kappa_{\bc' X'}$. 
\end{proposition}

Actually, tree $({\cal K}, {\cal E})$ is a {\it tree decomposition} for the $k$-tree graph. Tree decomposition is a technique developed in algorithmic structural graph theory which makes it possible to measure how much a graph is tree-like. In particular, a tree decomposition reorganizes a graph into a tree topology connecting subsets of the graph vertices (called {\it bags}, e.g., $(k+1)$-cliques in a $k$-tree) as tree nodes \cite{RobertsonAndSeymour1986,Bodlaender2006}. A heavy overlap is required between neighboring bags to ensure that the vertex connectivity information of the graph is not lost in such representation. However, finding a tree decomposition with maximum bag size $\leq k$ for a given arbitrary graph (of tree width $k$) is computationally intractable. Fortunately, for a $k$-tree or a backbone $k$-tree generated with the method presented in section 3, an optimal tree decomposition is already available according to Proposition~\ref{td}. 

Tree decomposition makes it possible to compute efficiently a large class of graph optimization problems on graphs of small tree width $k$, which are otherwise computationally intractable on restricted graphs. In particular, on the tree decompositions $({\cal K}, {\cal E})$ of a $k$-tree, many global optimization computation can be systematically solved in linear time  \cite{Bodlaender1988,ArnborgAndProskurowski1989,ArnborgEtAl1991,Courcelle1990}. 
The central idea is to build one dynamic programming table for every $(k+1)$-clique, and according to the tree decomposition topology, the dynamic programming table for a parent 
tree node is built from the tables of its child nodes. For every node (i.e., every $(k+1)$-clique), all (partial) solutions associated with the $k+1$ random variables are maintained, and (partial) solutions associated with variables only belonging to its children nodes are  optimally selected conditional upon the variables it shares with its children node. 
The optimal solution associated with the whole network is then present at the root of the tree decomposition. The process takes $f(k)n$ time, for some possibly exponential function $f$, where $f(k)$ is the time to build a table, one for each of the $O(n)$ nodes in the tree decomposition. For small or moderately large values of $k$, such algorithms scale linearly with the number of random variables in the network.

\section{Concluding Remarks}

We have generalized Chow and Liu's seminal work from Markov trees to Markov networks of tree width $\leq k$. In particular, we have proved that model approximation with Markov networks of tree width $\leq k$ has the minimum information loss when the network topology is a maximum spanning $k$-tree. We have also shown that learning Markov network topology of backbone $k$-trees can be done in polynomial time for every fixed $k$, in contrast to the intractability in learning Markov $k$-tree without the backbone constraint. This result also holds for a broader range of constraints. 

The backbone constraint stipulates a linear structure inherent in many Markov networks. Markov backbone $k$-trees are apparently suitable for modeling random systems involving the time series. In particular, we have shown the $k^{\rm th}$ order Markov chains are actually Markov backbone $k$-trees. The constrained model is also ideal for modeling systems that possess higher order relations upon a linear structure and has been successful in modeling 3-dimensional structure of bio-molecules and in modeling semantics in computational linguistics, among others.

\section*{Acknowledgement}

This research is supported in part by a research grant from National Institute of Health (NIGMS-R01) under the Joint NSF/NIH Initiative to Support Research at the Interface of the Biological and Mathematical Sciences.

\newpage

\section*{Appendix A}

\begin{theorem}
\label{appendix}
The probability function expressed in {\rm (\ref{joint_prob})} for Markov $k$-tree $G$ remains the same regardless the choice of creation order for $G$.
\end{theorem}

\begin{proof} 
Let $\Psi_{\bx}$ and $\Phi_{\bx}$ be two creation orders for $k$-tree $G=(\bx, E)$. For the same $k$-tree $G=(\bx, E)$, we define $P_{G}(\bx| \Psi_{\bx})$ and $P_{G}(\bx| \Phi_{\bx})$ to be respectively the joint probability distribution functions of $\bx$ under the topology graph $G$ with creation orders $\Psi_{\bx}$ and $\Phi_{\bx}$. Without loss of generality, we further assume that $\Phi_{\bx}$ be the creation order for $G$ as given in Definition~\ref{creation_order} of the following form:
\[ \bc_k X_{k+1} \bc_{k+1} X_{k+2} \dots \bc_{n-1} X_n \hspace{10mm} n\geq k\]

By induction on $n$, we will prove in the following the statement that $P_G(\bx | \Psi_{\bx}) = P_G(\bx | \Phi_{\bx})$.

\vspace{2mm}
\noindent
\underline{\it Basis}:  when $n=k$, 
\[P_G(\bx | \Psi_{\bx}) = P(X_1, \dots, X_k |\Psi_{\bx}) = P(X_1, \dots, X_k) = P(\bc_{k}) = 
P_G(\bx | \Phi_{\bx})\]

\noindent
\underline{\it Assumption}:  the statement holds for any $k$-tree of $n-1$ random variables. 

\vspace{2mm}

\noindent
\underline{\it Induction}:
Let $G=(\bx, E)$, where $\bx=\{X_1, \dots, X_n\}$, be a $k$-tree. Then it is not hard to see that subgraph $G'=(\bx', E')$ of $G$, after variable $X_n$ is removed,  is indeed a $k$-tree
for random variables $\bx'=\{X_1, \dots, X_{n-1}\}$. Furthermore, according to
Definition~\ref{creation_order}, 
$\Phi_{\bx'}$ is a creation order for $k$-tree $G'$.



Now we assume creation order $\Psi$ for graph $G =(\bx, E)$ to be
\[ {\bf D}_{k} Y_{k+1} {\bf D}_{k+1} Y_{k+2} \dots {\bf D}_{n-1} Y_{n}\]
where ${\bf D}_{k} = \{Y_1, \dots, Y_k\}$, and $Y_i \in \bx$ for all $i=1,\dots, n$.
Then there are two possible cases for $X_n$ to be in the creation order $\Psi_{\bx}$.

\vspace{2mm}
\noindent
{\bf Case 1}. $X_n \in {\bf D}_{k}$. 

Because $X_n$ shares edges only with variables in $\bc_k$, we have ${\bf D}_{k} \cup \{Y_{k+1} \} = \bc_{n-1} \cup \{ X_n\}$ and  ${\bf D}_{k+1} = {\bf D}_{k}\backslash \{X_n\} \cup \{Y_{k+1}\} = {\bf C}_{n-1}$. Therefore,  
\[\Psi_{\bx'} = {\bf D}_{k+1} Y_{k+2} \dots {\bf D}_{n-1} Y_{n}\] is a creation order for subgraph $G'=(\bx', E')$. Thus, based on (\ref{joint_prob}), we have 
\begin{equation}
\begin{split}
P_G(\bx |\Psi_{\bx}) & = P({\bf D}_{k}) P(Y_{k+1} | {\bf D}_{k}) \prod_{i=k+1}^{n-1} P(Y_{i+1} | {\bf D}_i, \Psi_{\bx'})\\
&=  P({\bf D}_{k}) P(Y_{k+1} | {\bf D}_{k}) 
\frac{P_{G'}(\bx' | \Psi_{\bx'})}{P({\bf D}_{k+1})}\\
& = P({\bf D}_{k} \cup \{Y_{k+1}\}) \frac{P_{G'}(\bx' | \Psi_{\bx'})}{P({\bf D}_{k+1})}\\
& = P({\bf C}_{n-1} \cup \{X_n\}) \frac{P_{G'}(\bx' | \Psi_{\bx'})}{P({\bf D}_{k+1})} \\
& = P({\bf C}_{n-1}) P(X_n | {\bc}_{n-1}) \frac{P_{G'}(\bx' | \Psi_{\bx'})}{P({\bf D}_{k+1})} \\
& = P(X_n | {\bc}_{n-1}) P_{G'}(\bx' | \Psi_{\bx'}) 
\end{split}
\end{equation}
By the assumption, $P_{G'}(\bx' | \Psi_{\bx'}) = P_{G'}(\bx' | \Phi_{\bx'})$, we have
\begin{equation}
\begin{split}
P_G(\bx |\Psi_{\bx}) 
& = P(X_n | \bc_{n-1}) P_{G'}(\bx' | \Phi_{\bx'}) \\
& = P(X_n | \bc_{n-1}) P(\bc_{k}) \prod_{i=k+1}^{n-1} P(X_{i} | {\bf C}_{i-1}, \Phi_{\bx'}) \\
& = P(\bc_{k}) \prod_{i=k+1}^{n} P(X_{i} | {\bf C}_{i-1}, \Phi_{\bx'}) \\
& = P_G(\bx | \Phi_{\bx})
\end{split}
\end{equation}

\vspace{2mm}
\noindent
{\bf Case 2}. $X_n = Y_m$ for some $m \geq k+1$.

Then ${\bf D}_{m-1} = {\bf C}_{n-1}$, and 
\[\Psi_{\bx'}  = {\bf D}_{k} Y_{k+1} \dots {\bf D}_{m-2} Y_{m-1}
{\bf D}_{m+1} Y_{m+1} \dots {\bf D}_n Y_n\]
is a creation order for subgraph $G'=(\bx', E')$. Based on (\ref{joint_prob}),
we have
\begin{equation}
\begin{split}
P_G(\bx |\Psi_{\bx}) & = P({\bf D}_{k}) \prod_{i=k+1}^n P(Y_{i} | {\bf D}_{i-1}, \Psi_{\bx})\\
&=  \bigg\{ P({\bf D}_{k}) \prod_{i=k+1}^{m-1} P(Y_{i} | {\bf D}_{i-1}, \Psi_{\bx'}) \prod_{i=m+1}^n P(Y_{i} | {\bf D}_{i-1}, \Psi_{\bx'})  \bigg\} \, P(Y_m | {\bf D}_{m-1}) \\
& = P_{G'}(\bx'|\Psi_{\bx'}) P(Y_m|{\bf D}_{m-1}) \\
& = P_{G'}(\bx'|\Psi_{\bx'}) P(X_n|\bc_{n-1})
\end{split}
\end{equation}
By the assumption, $P_{G'}(\bx'|\Psi_{\bx'}) = P_{G'}(\bx' | \Phi_{\bx'})$, we have
\begin{equation}
\begin{split}
P_G(\bx |\Psi_{\bx})
& = P_{G'}(\bx' | \Phi_{\bx'}) P(X_n|\bc_{n-1})\\
& = P(\bc_{k}) \prod_{i=k+1}^{n-1} P(X_i|\bc_{i-1}, \Phi_{\bx'}) \, P(X_n|\bc_{n-1})\\
& = P(\bc_{k}) \prod_{i=k+1}^{n} P(X_i|\bc_{i-1}, \Phi_{\bx})\\
& = P_{G}(\bx |\Phi_{\bx})
\end{split}
\end{equation}
\end{proof}

\end{document}